\documentclass[12pt]{article}
\usepackage[latin9]{inputenc}
\usepackage{geometry}
\usepackage{verbatim}
\usepackage{float}
\usepackage{amsmath}
\usepackage{amsthm}
\usepackage{amssymb}
\usepackage{graphicx}
\usepackage{setspace}
\usepackage[authoryear]{natbib}
\makeatletter

\DeclareTextSymbolDefault{\textquotedbl}{T1}

\theoremstyle{plain}
\newtheorem{thm}{\protect\theoremname}
\theoremstyle{plain}
\newtheorem{assumption}[thm]{\protect\assumptionname}
\theoremstyle{remark}
\newtheorem*{rem*}{\protect\remarkname}
\theoremstyle{plain}
\newtheorem*{prop*}{\protect\propositionname}
\theoremstyle{plain}
\newtheorem*{lem*}{\protect\lemmaname}

\AtBeginDocument{

}

\providecommand{\lemmaname}{Lemma}
\providecommand{\propositionname}{Proposition}

\usepackage{tikz}

\providecommand{\lemmaname}{Lemma}
\providecommand{\propositionname}{Proposition}

\providecommand{\lemmaname}{Lemma}
\providecommand{\propositionname}{Proposition}

\providecommand{\lemmaname}{Lemma}
\providecommand{\propositionname}{Proposition}

\providecommand{\assumptionname}{Assumption}

\providecommand{\theoremname}{Theorem}

\makeatother

\providecommand{\assumptionname}{Assumption}
\providecommand{\lemmaname}{Lemma}
\providecommand{\propositionname}{Proposition}
\providecommand{\remarkname}{Remark}
\providecommand{\theoremname}{Theorem}

\begin{document}
\title{On Vickrey's Income Averaging}
\author{\textbf{Stefan Steinerberger}\thanks{S.S. is supported by the NSF (DMS-1763179) and the Alfred P. Sloan
Foundation.}\textbf{ }\\
 {\small{}{}{}{}{}{}{}Yale University} \and \textbf{Aleh Tsyvinski}\\
 {\small{}{}{}{}{}{}{} Yale University}\thanks{We thank Hector Chade, Kjetil Storesletten, Georgii Riabov, Florian
Scheuer, and Philipp Strack for useful discussions.}}
\maketitle
\begin{abstract}
We consider a small set of axioms for income averaging -- recursivity,
continuity, and the boundary condition for the present. These properties
yield a unique averaging function that is the density of the reflected
Brownian motion with a drift started at the current income and moving
over the past incomes. When averaging is done over the short past,
the weighting function is asymptotically converging to a Gaussian.
When averaging is done over the long horizon, the weighing function
converges to the exponential distribution. For all intermediate averaging
scales, we derive an explicit solution that interpolates between the
two. 
\end{abstract}
\newpage{}

\section{Introduction}

How to average over the past is one of the most basic questions that
arises in a variety of economic fields. Our particular focus is on
a classic public economic issue -- how to average income for the
tax purposes -- but the answer is broadly applicable to many other
topics. The question of income averaging is typically attributed to
Vickrey (1939) and can be summarized as follows.\footnote{See also Simons (1938).}
Consider a progressive tax system and two taxpayers with the same
income over a period of time. The person with the constant income
pays a lower total amount of taxes than the person with the fluctuating
income. Averaging income to equalize the tax burden then may be desirable.\footnote{A related issue is the choice of the reference period for taxation.
Most commonly the taxes are assessed on the annual basis. However,
in principle, a government may use a shorter or a longer accounting
period. } 

In this paper, we want to abstract from the desirability of averaging.
Our goal is to formalize the question of averaging and to propose
a small set of reasonable axioms that an averaging mechanism may possess.
Given this set of the axioms, we are interested in what averaging
functions may arise. More broadly, the question of averaging over
the past, given a small set of assumptions, may be of use in a variety
of other applications such as behavioral economics or dynamic contracting.

We chose two main axioms that the averaging function that weights
income at different periods should satisfy. The first is recursivity
-- averaging should treat various scales of smoothing in a unified
way. In this context, the assumption implies that averaging at some
scale and then at another scale is equivalent to averaging at the
combined scale. This condition ensures that all these outcomes agree:
there is no difference in averaging over, say, a year, twelve units
of a month or 52 weeks. Alternatively, one can think of this assumption
as stating that no scale of averaging is singled out and all of them
are treated equally. It is natural that a reasonable multiscale averaging
method has the scales connected with this intrinsic compatibility
condition. The second is a continuity or locality assumption that
requires that the very distant incomes do not play a disproportionate
role in averaging. We also need to postulate how we treat the present,
that is, to set the boundary condition of the averaging rule. In addition,
we impose some other, less essential conditions, such as various normalizations.

In principle, many weighting functions are possible: equal weights
for all income, assigning lower weights to the more distant past,
singling out some incomes, etc. We show that the small number of assumptions
that we make result in a definite general weighting scheme. Specifically,
the smoothed incomes and the weighting function are a solution to
the second order parabolic partial differential equation. The easiest
way to describe the intuition behind this result is using a probabilistic
argument. One can think of averaging being done using a stochastic
process. Fix a given time and consider a stochastic process originating
in that period. The probability of the process reaching some other
time is then the weight that the averaging assigns to that income.
The recursivity (or semi-group) assumption implies that such a process
is Markovian. The locality assumption implies that the process has
continuous paths. The classic result is that the process is a diffusion.
Since we are averaging over the past and given the behavior on the
boundary, the resulting process is then a reflected Brownian motion
with a drift. The weighting function is the density of this process.
The smoothed incomes in turn satisfy a second order parabolic partial
differential equation.\footnote{In fact, we do not have to run average on the pre-tax income and could
rather consider smoothing or averaging the tax contributions or the
post tax income directly.}

The density of such process that defines a weighting function is known
in a simple closed form. From the point view of present, the averaging
function has particularly meaningful properties. Consider averaging
over the very short period in the past. In this case, the drift does
not have any substantial effect and the averaging is done mainly via
the normal density with the nonzero mean determined by the drift.
Consider now averaging over the long horizon. A remarkable fact in
probability theory is that for any constant positive drift, the weighting
function converges to an exponential. For all the averaging scales
in between these two, there is an explicit solution that interpolates
between them. We also show that the smoothed income has a particular
strong smoothing structure of a gradient flow.

Finally, we want to remark that our work connects the economic question
of averaging to two literatures that previously have not been represented
in economics. The first is the mathematical imaging and vision literature
that considers representation of images over various smoothing scales
(see, e.g. Aubert and Kornprobst 2006 and Lindeberg 2013). The second,
is the Sch{\"o}nberg's theory of variation diminishing transformations
and Polya's frequency functions (Sch{\"o}nberg 1948, Steinerberger 2019).\footnote{The literature on total positivity that builds on this work (Karlin
1968) is used more extensively in economics.} Our results nest and interpolate between the results on averaging
obtained in these two approaches.

\section{Literature}

The question of Vickrey's income averaging is classic and familiar
to any student of public finance. While this mechanism is not widely
used in fiscal practice today (there are some provisions for income
averaging for artists and farmers), it was extensively employed in
the past. Great Britain applied a progressive tax schedule to the
average of the individual income of the previous three years from
1799 to 1926. Between 1923 and 1938 Australia used a five-year moving
average of income over the past five years (Holt 1949). Gordon and
Wen (2018) describe a more recent experience the summary of which
we present: The United States introduced general income averaging
in 1964 and it was repealed in the 1986 Tax Reform Act. In Canada,
a similar policy to that in the United States was introduced in 1972
together with forward averaging of the income-averaging annuity contracts
and had been in effect until 1988. The impact of progressive tax rates
on realized capital gains was one reason for setting of low tax rates
on capital gains. There are also several prominent recent proposals
to reintroduce income averaging. Batchelder (2003) proposes targeted
averaging for the poor in the context such as EITC in the US. In Canada,
Mintz and Wilson (2002) for primarily retirement savings plans and
Gordon and Wen (2018) more broadly argue for income averaging. 

The question of averaging regularly appears in models of dynamic taxation
in which agents experience idiosyncratic shocks. The Mirrlees review
(Mirrlees and Adam 2010) that analyzes the theoretical foundations
behind the practical tax design devotes a signficant amount of space
to the question of lifetime earning variability and lifetime income
as the tax basis. Diamond (2009) discusses Vickrey's income averaging
in the context of design of pensions systems. Huggett and Para (2010)
study optimal lifetime tax in a model of social insurance. Kapicka
(2017) extends the model of Heathcote, Storesletten, and Violante
(2014) to allow for history-dependent taxes. He finds that the optimal
weights on past incomes decline geometrically at a rate equal to the
discount rate and that the gains from the history dependence are large.
Jacobs (2017) reviewing digitization and taxation states increased
feasibility of practical implementation of the income averaging rules.

There are relatively few recent studies of the empirical effects of
income averaging.\footnote{In terms of the assessment of the practical implementation of income
averaging, there is an extensive literature in law (see, e.g., a summary
in Buchanan 2005).} The most comprehensive is Gordon and Wen (2018) which also contains
a review of the older literature on the topic. For the Canadian data,
they find that while fluctuation penalty is low on average, 10 percent
of taxpayers faced annual tax penalties of 1 percentage point of their
income and 1 percent of taxpayers paid 4 percentage points. That is,
those in the top 1 percentile of the penalty paid 4 percent of their
average annual income more in taxes per year if they had been able
to perfectly average. The highest percentile is composed largely of
the self-employed or those with the realized capital gains. What is
more, 57 percent of taxpayers located between the 95th and 100th percentiles
of the penalty are from the bottom income quintile. Similarly, in
the US, Batchelder (2003) finds that families the bottom quartile
of families ranked by the annual income faced an additional effective
tax rate of 2.0 percentage points higher under annual income measurement
than it would be if income were fully averaged, whereas for the top
quartile's rate it is only 0.5 percentage points higher. Bargain,
Trannoy, and Pacifico (2017) examine French administrative data and
show that increasing the tax frequency can lead to substantial social
welfare gains, coming to an important degree from the bottom of the
income distribution.\footnote{Saez (2002) considers the question for understanding the optimal period
for computing the time liability.}

There are two related answers in the mathematical literature to the
question of averaging. 

The first is given by the literature on scale-space in mathematical
imaging and vision analysis (see e.g., books by Aubert and Kornprobst
2006 and Lindeberg 2013). This literature studies image representations
at various scales -- from the finest scale that represents the original
image to the coarser scales of the smoothed versions of the images.
Smoothing is conducted at various scales which are tightly related
to each other. The analysis there is mainly concerned with two-sided
averages and derives a deep and substantial result -- a Gaussian
kernel arises as a unique averaging object based on a small set of
assumption when averaging is done over the whole line. The Gaussian
kernel has many properties and appears in a variety of fields of mathematics:
in the scale-space literature it is derived primarily using two main
assumptions: varying forms of smoothing and the semigroup (recursivity)
structure.\footnote{The literature identified also some other possible results where under
different assumptions one may get the stable distributions from probability
theory (see e.g., Pauwels, Van Gool, Fiddelaers, and Moons 1995) or
a nonlinear diffusion (see e.g., Alvarez, Guichard, Lions, and Morel
1993) as the averaging principle -- yet the Gaussian is a canonical
and most widely used answer.} For example, Lindeberg (1997) writes ``A notable coincidence between
the different scale-space formulations that have been stated is that
the Gaussian kernel arises as a unique choice for a large number of
different combinations of underlying assumptions (scale-space axioms).''
This is very important point: given a variety of reasonable assumptions
you make (there is quite a lot of natural things one could assume),
you usually end up with the Gaussian. At the same time, the one-sided
question of averaging over the past is much less studied and the answer
is less canonical in that literature.\footnote{See Fagerstr{\"o}m (2005), Lindeberg and Fagerstr{\"o}m (1996), Lindeberg
(2011), and Salden, ter Haar Romeny, and Viergever (1998).}

The second literature (Steinerberger 2019) considers one-sided averages
without the semi-group property. The main insight is based on the
Sch{\"o}nberg's theory of variation diminishing transformations and Polya's
frequency functions (Sch{\"o}nberg 1948). The variation diminishing property
(total positivity) is a particular way to define uniform smoothing
at all scales, and it requires that the number of the function's crossings
at any levels is decreased.\footnote{Karlin (1968) shows the unique averaging kernel that is variation
diminishing and is a semigroup on the whole line is the Gaussian.} The second assumption that is imposed is monotonicity where the more
recent past is weighed heavier than the more distant path. Steinerberger
(2019) shows that total positivity and monotonicity on the half line
lead to the unique weighting given by the exponential distribution.
That is, the ``exponential smoothing'' classical in time series
analysis (Brown 1957 and Holt 1957 ) arises from a small set of axioms. 

Our result can generate the conclusions of both of these approaches
from a reasonable set of assumptions as well as a range of intermediate
results. This leads to the weighting scheme that interpolates between
the Gaussian (similar in spirit to the scale-space theories) for a
short horizon of averaging and the exponential smoothing for a positive
drift case if one considers the long time horizon.

\section{A question of averaging}

In this section, we define a question that we aim to study. Let a
bounded measurable income function $f:\mathbb{R}\rightarrow\mathbb{R}$
be defined on time $x\in$$(-\infty,\infty)$. That is, $f\left(x\right)$
is income at time $x$. We are interested, at a given time, in finding
the average of the income in the past. We want this process to be
translation invariant: the way we average over the past should not
depend on whether it is, for example, currently January or July. Moreover,
we would like the process to be linear in the income: the sum of two
averaged incomes should be the average of the sum of the two incomes.
The canonical setting for this is to average by
\[
g(x)=\int_{0}^{\infty}f(x-y)h(y)dy,
\]

where $h:[0,\infty)\rightarrow\mathbb{R}$ is a (not necessarily continuous)
weighting function, assigning weight $h\left(y\right)$ to the income
$y$ units of time in the past. Many different weighting functions
are conceivable. For example, one could simply average the income
incurred over the last $a$ units of time - this would correspond
to $h$ being a step function on $[0,a]$ having constant value $1/a$.

Our main question is what kind of averaging functions would arise
given a small number of reasonable axioms. We emphasize that while
we give one possible answer to this question under reasonable axiomatic
assumptions, we believe that this question is well worth of further
study. In particular, it is quite conceivable that other sets of assumptions
would lead to other natural functions $h(y)$. We recall that in the
mathematical imaging literature, the Gaussian arises naturally from
a wide variety of very different assumptions. No such analogous way
of forming averages seems to be known for the cases of one-sided averages;
both the one-sided Gaussian and the exponential distribution appear
in very different settings but the existing literature is very sparse
and not as comprehensive as the two-sided case. We believe this to
be an exciting avenue for further work. 

The second question, that we are interested in, is the issue of scale.
We would like the weighting function and averaging to apply at different
scales. That is, we want the function $h$ to be in fact a family
of functions parameterized by a scale parameter $t$. Intuitively,
the scale parameter measures both the range of averaging (the effective
length of the time period of averaging) and, as it will turn out,
the degree of smoothing.\footnote{That is, considering longer intervals allows for overall smoother
averages -- if I take my daily income, it may fluctuate significantly;
if I average over a week, it will be smoother; and if I average over
a year it will be smoother still.} The reason for considering this parameter is that income fluctuations
may occur at different time scales -- from the weekly earnings of
a restaurant worker to multi-year royalties of a songwriter. Having
a range of scale parameters allows to consider different averaging
requirements that varying circumstances necessitate. If no assumptions
are made on the range of income fluctuations a priori then all scales
should be considered simultaneously. A challenge is to model and understand
the representation of averaging not as an unrelated rules for different
scales but rather to have a unifying principle that operates at all
scales.\footnote{This is exactly the foundations of the scale-space theory in mathematical
imaging and vision where an image has to be represented at different
scales simultaneously -- from the minute details at the close inspection
to the outlines of the main features when viewed from a distance (see,
e.g., Koenderink 1994 and Lindeberg 2013).}

Without loss of generality, we fix the initial time to $0$ and call
it the present. Let $x\in(-\infty,0]$ denote some time in the past
and $f(x)$ denote income at time $x$. We further introduce another
parameter -- scale $t$. We are interested in the transformations
of the income function $f\left(x\right)_{x\in(-\infty,0]}$ at different
scales $t$: $u\left(t,x\right)$. The function $u\left(t,x\right)$
is the smoothed income at time $x$ where the scale of smoothing is
$t$. Specifically,

\begin{equation}
u(t,x)=\int_{-\infty}^{0}f(y)p_{t}(x,y)dy,\label{eq:averaging}
\end{equation}
where $p_{t}(x,y)\ge0$ is an averaging or weighting function.\footnote{Instead of making the assumptions on the form of (\ref{eq:averaging})
we could have more abstractly considered a family of linear operators
$T_{t}$ acting on bounded and continuous functions $f\left(x\right)_{x\in(-\infty,0]}$.
Riesz-Markov-Kakutani representation theorem implies that a (positive)
continuous linear functional $f\rightarrow T_{t}f$ is represented
by a measure $T_{t}f\left(x\right)=\int_{-\infty}^{0}f\left(y\right)P_{t}\left(x,dy\right)$.
Further assuming $T_{t}1=1$ implies that $P_{t}\left(x,dy\right)$
is a probability measure.} 

The scale $t$ in what follows also determines the intensity of smoothing.
For a given scale $t$, this operation takes an initial income function
$f$ and transforms it into a new function $u(t,x)_{x\in(-\infty,0]}$
by convolving with the function $p_{t}$. That is, for a given $x$,
$u(t,x)$ is a weighted average of incomes $f(y)$ with the weights
$p_{t}(x,y)$. We are particularly interested in the value of smoothed
income at the present time $x=0$ at various scales $t$:
\[
g(0)=u(t,0)=\int_{-\infty}^{0}f(y)p_{t}(0,y)dy.
\]

\section{Main Assumptions}

At this stage, the weighting function $p_{t}(x,y)$ and the corresponding
smoothed income $u(t,x)$ can be very general, and we now state further
assumptions that allow us to specifically determine it.

A natural condition, that is often not even mentioned, is that if
the income function is constant, $f(x)\equiv c$, then the averaged
income function should also be constant and equal to the same numerical
value. We also normalize the weighting function: for all $x,t$ 
\[
\int_{-\infty}^{0}p_{t}(x,y)dy=1.
\]

The natural averaging after 0 units of time have passed is to simply
return the original value at the point since nothing has yet happened.
We assume then the initial condition that $p_{0}(x,\cdot)=\delta_{x}$,
where $\delta_{x}$ is the Dirac delta function.
\begin{assumption}
\textbf{\label{assu:=00005BSemi-group=00005D}{[}Recursivity{]}} For
any $x,y\in(-\infty,0]$ and $t,s\ge0$, 
\[
p_{t+s}(x,y)=\int_{-\infty}^{0} p_{s}(x,z)p_{t}(z,y)dz.
\]
\end{assumption}

This is a natural assumption that connects different scales of averaging.\footnote{More abstractly, we could have posed the semi-group property of the
operator $\left(T_{t}\right):T_{0}f=f,T_{s}\circ T_{t}=T_{s+t},$
for all $s,t\ge0$.} It is similar to many other recursive formulations common in economics.
In this context, the assumption implies that averaging at the scale
$s$ and then at the scale $t$ is equivalent to averaging at scale
$t+s$. Alternatively, one can think of this assumption as stating
that no scale of averaging is singled out and all of them are treated
equally. One could thus interpret it as a statement about the internal
consistency: if such an averaging method were to be implemented, then
a citizen could conceivably ask to have their income averaged twelve
times over the scale of a month as well as over the scale of a year
and then pick the more favorable outcome. This condition ensures that
all these outcomes agree: there is no difference in averaging over
a year, twelve units of a month or 52 weeks. It is natural that a
reasonable multiscale averaging method has the scales connected with
this intrinsic compatibility condition.

We now turn to the third assumption -- locality.
\begin{assumption}
\label{assu:=00005BLocality=00005D}\textbf{{[}Locality{]}} For each
$x\in(-\infty,0)$ and $\varepsilon>0$
\[
\int_{|y-x|>\varepsilon}p_{t}(x,y)dy=o(t).
\]
 Furthermore, there exist the infinitesimal characteristics $a$ and
$r$:
\[
\int_{|y-x|\le\varepsilon}(y-x)p_{t}(x,y)dy=rt+o\left(t\right),
\]
\[
\int_{|y-x|\le\varepsilon}(y-x)^{2}p_{t}(x,y)dy=at+o\left(t\right).
\]
\end{assumption}

In essence, this assumption states a form of continuity for the averaging
operator that for a given $x$, only the local values $y$ matter
for the resulting average (for small $t$). There is also an additional
assumption built in here -- time and scale independence of the coefficients
-- which we chose not to state separately. We could have instead
assumed that $a(t,x)$ and $r(t,x$) with the results straightforwardly
extending.\footnote{More broadly, it may be of interest to also incorporate some assumptions
related to the time-value of money which would determine the value
of $r(x)$. } There is also one symmetry and, without loss of generality, we can
set $a=1$. \footnote{Different $a$ would correspond to speeding up the time; however,
since time will actually be one of the parameters in our solution
formula, it can be recovered from there.}

There is also a probabilistic interpretation of this assumption. One
can think of the weight $p_{t}\left(x,y\right)$ as the probability
that a process travels from $x$ to $y$ in time $t$. The first part
of the above assumption then states that the probability of non-local
jumps is vanishingly small in time. This assumption together with
the recursivity assumption then ensures the continuity of the paths
of the stochastic process (see, e.g., Feller 1954).

Finally, we need an assumption on the behavior of the weighting function
$p_{t}(x,y)$ at the boundary. There are two canonical ways of dealing
with a boundary: to impose Dirichlet or Neumann conditions. Dirichlet
conditions are not suitable for our application because we would not
want to impose a zero weight being assigned to the present income
at $x=0$.\footnote{Also Dirichlet condition contradicts the fact that the total mass
is preserved.} This leaves Neumann conditions as a natural choice:\footnote{There are also Robin conditions that one could impose, these are of
the form 
\[
p_{t}(x,y)+\alpha\frac{\partial}{\partial x}p_{t}(x,y)\big|_{x=0}=0\qquad\mbox{for some fixed}~\alpha\in\mathbb{R}.
\]
However, since this also involves the value of the function at the
boundary (the quantity of interest), it seems unnatural to force it
to be of any particular form.}
\begin{assumption}
\label{assu:=00005BBoundary-Conditions=00005D}\textbf{{[}Boundary
Conditions{]}} For any $y<0$, 
\[
\frac{\partial}{\partial x}p_{t}(x,y)\big|_{x=0}=0.
\]
\end{assumption}

\begin{rem*}
Finally, it is important to note that clearly there are many possible
choices for the assumptions on averaging. One goal of our work is
to formalize the question and open the venues to exploring possibly
other choices of assumptions.
\end{rem*}

\section{Results}

We now state the result characterizing the weighting function $p_{t}\left(x,y\right)$. 
\begin{prop*}
The evolution equation that follows from Assumptions \ref{assu:=00005BSemi-group=00005D}-\ref{assu:=00005BBoundary-Conditions=00005D}
is 
\begin{equation}
\frac{\partial}{\partial t}p_{t}(x,y)=\frac{1}{2}\frac{\partial^{2}}{\partial x^{2}}p_{t}(x,y)+r\frac{\partial}{\partial x}p_{t}(x,y).\label{eq:diffusion}
\end{equation}
Furthermore, $p_{t}(x,y)$ has an explicit closed form given by the
probability distribution function of the Reflected Brownian Motion:
\begin{equation}
p_{t}(x,y)=2re^{2ry}\Phi\left(\frac{rt+x+y}{\sqrt{t}}\right)+\frac{1}{\sqrt{t}}\phi\left(\frac{-rt-x+y}{\sqrt{t}}\right)+\frac{e^{2ry}}{\sqrt{t}}\phi\left(\frac{rt+x+y}{\sqrt{t}}\right),\label{eq:Kernel}
\end{equation}
where $r\in\mathbb{R}$, $x,y<0$, $\phi$ is the probability density
function of the standard $\mathcal{N}(0,1)$ Gaussian distribution,
and $\Phi$ is its cumulative density function.
\end{prop*}
The first part of the result follows from the classic paper of Kolmogorov
(1931) on the connection of the diffusion processes with the second
order parabolic partial differential equations. The assumptions of
the recursivity (semigroup) and locality (continuity) assure that
the associated process is a diffusion with the density characterized
by (\ref{eq:diffusion}).\footnote{In fact, one would need a weaker set of assumptions to ensure that
$p_{t}(x,y)$ is represented by a second-order differential equation.
Specifically, assuming that the operator $T_{t}$ is a semi-group,
preserves unity ($T_{t}1=1)$, and is non-negative ($f\ge0\Rightarrow T_{t}f\ge0)$
would suffice. This can be proven modifying Stroock (2008, Lemma 1.1.6,
p.2) and is a consequence of a more general result of Peetre (1959)
that local operators are differential operators of finite order.} 

This particular partial differential equation \ref{eq:diffusion}
is actually quite simple and is easy to solve on the whole line $\mathbb{R}$.
What is different in our setting is that we are working on the half-line
$\mathbb{R}_{\leq0}$ and have reflecting boundary conditions which
is more challenging. For any fixed $x<0$ and any $t>0$, we can interpret
$p_{t}(x,\cdot)$ as a probability distribution function. This probability
distribution describes the position of a Brownian particle started
in $x$ and moving with constant drift in direction $r$ (which points
towards the origin for $r>0$ and away from the origin for $r<0$).
Specifically, the relevant process $Z_{t}$ is defined as follows
-- this is the Skorohod reflection problem (Skorohod 1961, 1962).
Let $(B_{t})_{t\geq0}$ be a Brownian motion, and consider the process
\[
X_{t}=x+rt+B_{t}.
\]
There exists a unique increasing continuous function $L_{t}$ such
that $L_{0}=0,$ $Z_{t}=X_{t}-L_{t}\leq0$ and $L_{t}$ grows only
at the points where $Z_{t}=0.$ Precisely, $L_{t}=\sup_{0\leq s\leq t}X_{s}^{+}.$

The second part of the result and the explicit form of weighting $p_{t}(x,y)$
function in (\ref{eq:Kernel}) follows from the results in the queueing
theory of Harrison (2013, p. 48) and Glynn and Wang (2018). 

Moreover, the smoothed income at scale $t$ 
\[
u(t,x)=\int_{-\infty}^{0}{p_{t}(x,y)f(y)dy}
\]
is the solution $u(t,y)$ of the initial-boundary value problem 
\[
\frac{\partial}{\partial t}u(t,x)=\frac{1}{2}\frac{\partial^{2}}{\partial x^{2}}u(t,x)+r\frac{\partial}{\partial x}u(t,x)\qquad\mbox{in}~(-\infty,0),
\]
\[
\frac{\partial}{\partial x}u(t,0)=0;u(0,x)=f(x).
\]

In particular, the primary object of our interest -- smoothed income
at the present time ($x=0)$ at scale $t$ is given by

\[
g(0)=u(t,0)=\int_{-\infty}^{0}{p_{t}(0,y)f(y)dy}
\]
This is the setting that we originally set out to study, and we have
identified an averaging function $h(t,y)=p_{t}(0,y)$.

We now derive two important properties of the probability distribution
function function $p_{t}\left(x,y\right)$ -- the behavior at the
small and large scales.
\begin{lem*}
We have:

(1) for any fixed $x<0$ and $y<0$
\[
p_{t}\left(x,y\right)\sim\frac{1}{\sqrt{t}}\phi\left(\frac{y-x-rt}{\sqrt{t}}\right),t\to0.
\]

(2) if $r>0$, then, for all $x<0$ and $y<0$, we have 
\[
\lim_{t\rightarrow\infty}p_{t}(x,y)=2re^{2ry},
\]

$ $ $ $ $ $ $ $ if $r<0$, then $p_{t}(x,\cdot)$ converges to
0 on every compact interval as $t\rightarrow\infty$.
\end{lem*}
\begin{proof}
\ 

Part (1): Smoothing at small scale: $t\rightarrow0$. 

For any $x,y<0$ consider the ratio
\[
\frac{p_{t}\left(x,y\right)}{\frac{1}{\sqrt{t}}\phi\left(\frac{y-x-rt}{\sqrt{t}}\right)}=1+\frac{\sqrt{t}2re^{2ry}\Phi\left(\frac{y+x+rt}{\sqrt{t}}\right)}{\phi\left(\frac{y-x-rt}{\sqrt{t}}\right)}+e^{2ry}\frac{\phi\left(\frac{y+x+rt}{\sqrt{t}}\right)}{\phi\left(\frac{y-x-rt}{\sqrt{t}}\right)}.
\]
We observe that, as $t \rightarrow 0$,
\begin{align*}
\frac{\phi\left(\frac{y+x+rt}{\sqrt{t}}\right)}{\phi\left(\frac{y-x-rt}{\sqrt{t}}\right)}&=\exp\left(\frac{(y-x-rt)^{2}-(y+x+rt)^{2}}{2t}\right)\\&=\exp\left(\frac{-2y(x+rt)}{t}\right)\to0
\end{align*}
Also, as $t \rightarrow 0$,
\begin{align*}
\frac{\Phi\left(\frac{y+x+rt}{\sqrt{t}}\right)}{\phi\left(\frac{y-x-rt}{\sqrt{t}}\right)} &=e^{\frac{(y-x-rt)^{2}}{2t}}\int_{-\infty}^{\frac{y+x+rt}{\sqrt{t}}}e^{-u^{2}/2}du\\
&\leq\sqrt{\frac{\pi}{2}}\exp\left(\frac{(y-x-rt)^{2}-(y+x+rt)^{2}}{2t}\right)\to0.
\end{align*}
So, for all $x,y<0$ and $t\rightarrow 0$,
\[
\frac{p_{t}\left(x,y\right)}{\frac{1}{\sqrt{t}}\phi\left(\frac{y-x-rt}{\sqrt{t}}\right)}\to1.
\]
and the equivalence $p_{t}\left(x,y\right)\sim\frac{1}{\sqrt{t}}\phi\left(\frac{y-x-rt}{\sqrt{t}}\right),t\to0,$
is established. 

Part (2): Smoothing at large scale: $t\rightarrow$$\infty$. This
result shows that for the positive drift (towards the origin) at large
scale $t\rightarrow\infty$, the function $p_{t}(x,y)$ converges
to a universal (not depending on time $x$) limiting object 
\[
p_{t}(x,y)\rightarrow2re^{2ry}
\]
which is a negative of the exponential distribution. We consider the
steady-state for the Kolmogorov forward equation: 
\[
\frac{1}{2}f''(y)=rf'(y).
\]
Clearly, 
\[
f(y)=Ae^{2ry}+B
\]
Boundary condition $f'(0)=2rf(0)$ implies that $B=0.$ Normalization
to the mass equal to one gives $A=2r.$
\end{proof}
The lemma above has the following meaning. Part (1) considers averaging
over small time, that is, over a very short effective range. The main
idea is that we average over very short windows of time, the boundary
condition has no effect, the drift is still presented and we get averaging
with the (non-centered) Gaussian distribution). Part (2) is in fact
a rather remarkable fact in probability theory. The dynamical
situation is as follows: we have a Brownian particle on $\mathbb{R}_{\leq0}$
that is reflected at the origin. A particle such as this would slowly
start exploring the space and be spread out more and more (roughly
at scale $\sim\sqrt{t}$ after $t$ units of time). Without the drift
$r$ (or with the drift away from zero, $r<0)$, there is no interesting
limit as $t\rightarrow\infty$, the probability distribution function
of Brownian motion goes to 0 (because the particles are spread out
more and more). Here (when $r>0$), we have a slightly different situation
resulting in a very different outcome: we have a constant drift (of
strength $r$) moving the particles back to the origin. As time becomes
large there is a limiting profile resulting as the balance of two
forces: the constant drift trying to move everything to the origin
and the Brownian particle moving around randomly. This limiting profile
is given by the exponential distribution (which we encountered previously
for very different reasons). The parameter $r>0$, drift towards the
origin, can be thought of modeling a form of monotonicity where from
the point of view of the present $x=0$ the more recent observations
receive a higher weight.

The explicit form for the density $p_{t}(x,y)$ in equation (\ref{eq:Kernel})
thus gives an interpolation of the weighting function between the
Gaussian and the exponential distribution. We plot it in Figure \ref{fig:Extrapolation}.

\begin{figure}[H]
\begin{centering}
\includegraphics[width=0.6\columnwidth]{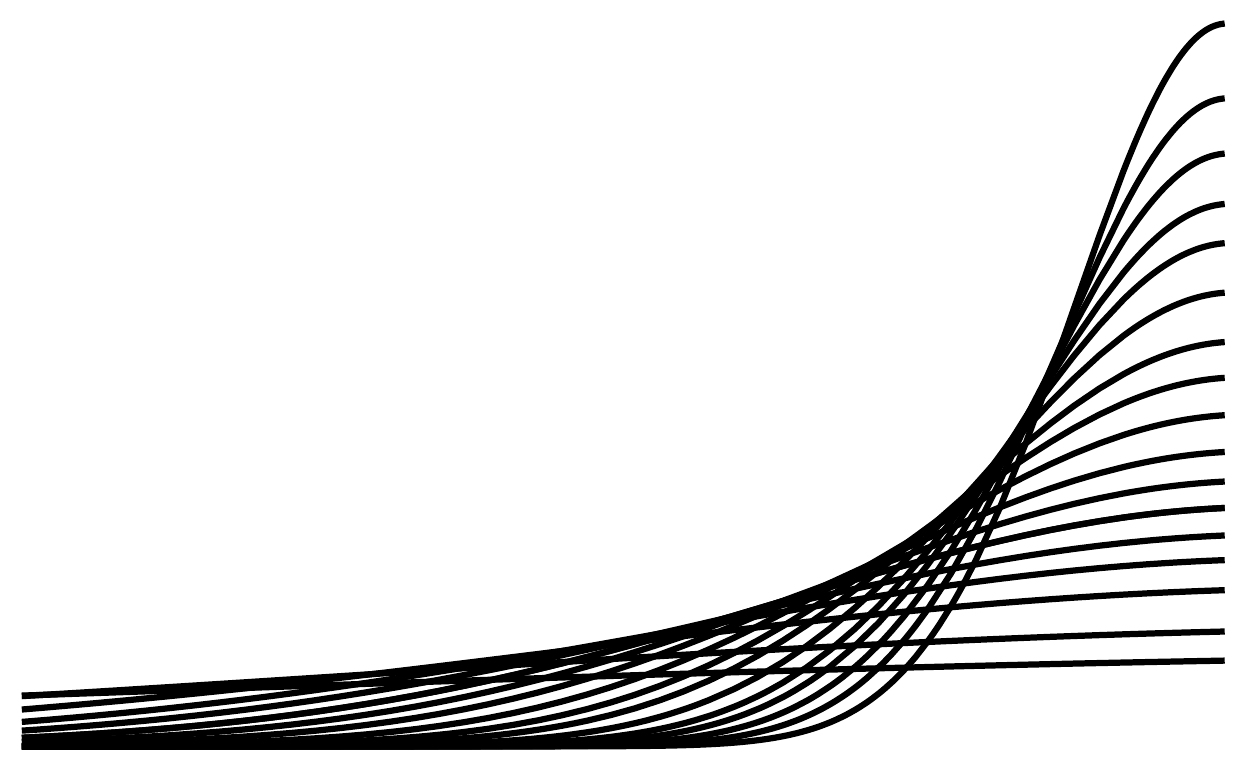}
\par\end{centering}
\caption{\label{fig:Extrapolation}Weighting function $p_{t}(0,y)$ interpolates
between Gaussian and exponential distributions ($r>0$).}
\end{figure}

Figure \ref{fig:Smoothing} plots smoothed income $g(x)$ at different
scales $t$. Higher scales imply a larger effective range of averaging
and result in smoother profile of income. 

\begin{figure}[H]
\begin{centering}
\includegraphics[width=0.7\columnwidth]{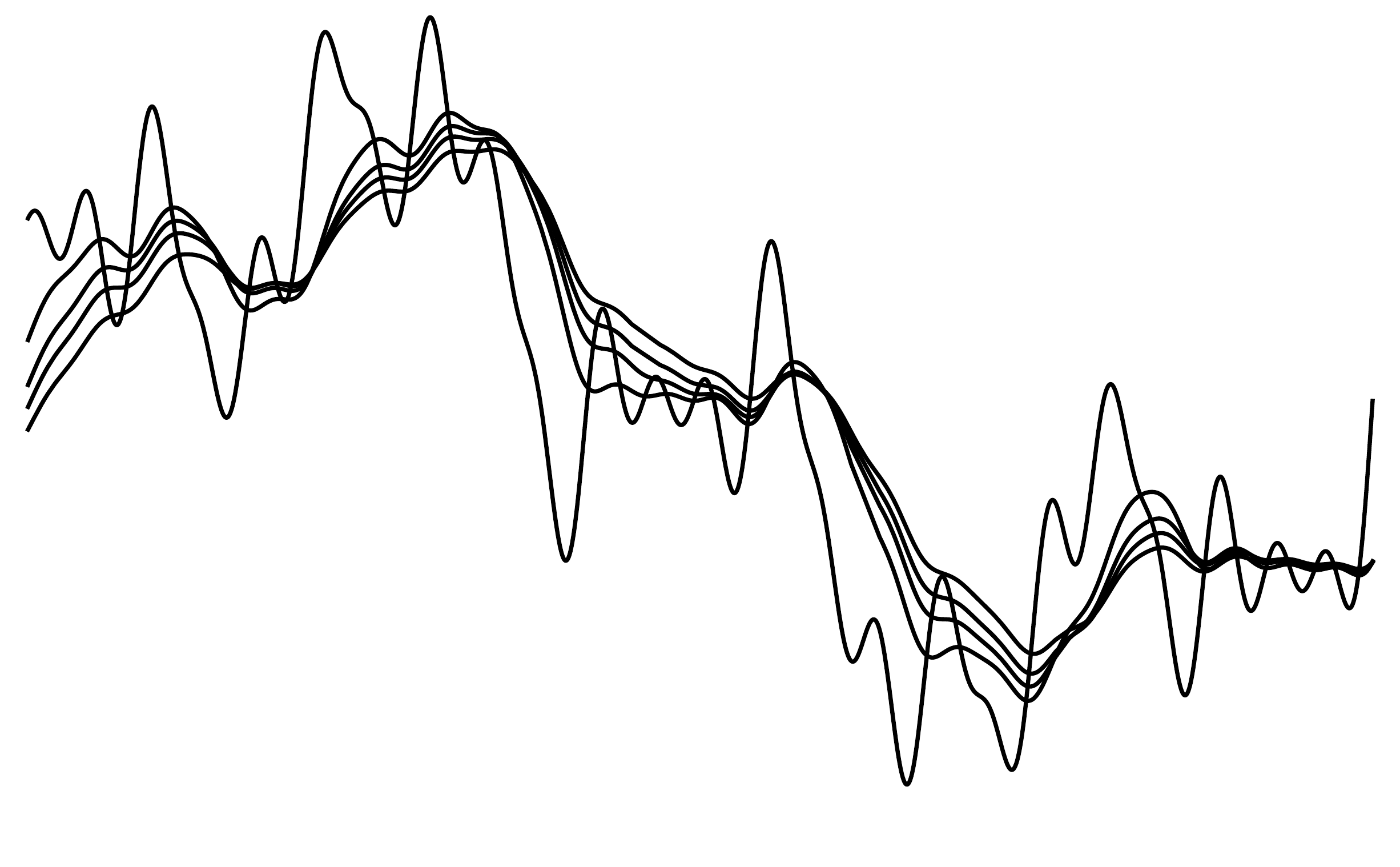}
\par\end{centering}
\caption{\label{fig:Smoothing}Smoothed income, $g(x)$, at different scales
$t$.}
\end{figure}

Finally, we are interested in the smoothing properties of the weighting
function that we found. Specifically, we are interested in knowing
whether our equation 
\[
\frac{\partial}{\partial t}u\left(t,x\right)=\frac{1}{2}\frac{\partial^{2}}{\partial x^{2}}u\left(t,x\right)+r\frac{\partial}{\partial x}u\left(t,x\right)
\]
 arises as a gradient flow on some space. The gradient flow is the
analogue of the usual gradient descent process but for the space of
the functions. That is, we evolve the whole function in the direction
of the steepest increase in some objective function.\footnote{See Steinerberger and Tsyvinski (2019) for a detailed description
of gradient flows arising in the context of optimal taxation.} 

Consider the Hilbert space $L^{2}((-\infty,0],\mu),$ where $\mu(dx)=e^{2rx}dx.$
Let $M$ be the subspace consisting of all continuously differentiable
functions $u:(-\infty,0]\to\mathbb{R},$ such that $u'(0)=0.$ On
$M$ we define a functional 
\[
I(u)=\int_{-\infty}^{0}u_{x}^{2}d\mu,
\]
where $\mu=e^{2rx}dx$. 

The corresponding gradient flow is a function $u:[0,\infty)\to M$
such that $\partial_{t}u=-\nabla I(u),$where $\nabla I$ is the gradient
of the functional $I.$ We compute the gradient:

\[
I(u+\epsilon v)-I(u)=\frac{\epsilon}{2}\int_{-\infty}^{0}u'(x)v'(x)e^{2rx}dx+o(\epsilon).
\]
So,
\begin{align*}
(\nabla I(u),v)&=\frac{1}{2}\int_{-\infty}^{0}u'(x)v'(x)e^{2rx}dx\\
&=-\frac{1}{2}\int_{-\infty}^{0}v(x)(u''(x)+2ru'(x))e^{2rx}dx.
\end{align*}

It follows that 
\[
\nabla I(u)=-\frac{1}{2}u''(x)-ru'(x),
\]
and the gradient flow coincides with the equation 
\[
\frac{\partial}{\partial t}u\left(t,x\right)=\frac{1}{2}\frac{\partial^{2}}{\partial x^{2}}u\left(t,x\right)+r\frac{\partial}{\partial x}u\left(t,x\right).
\]

Let us consider $u(t+\varepsilon,x)=u(x)+\varepsilon v(x)$. We then
construct the flow in the direction opposite to the gradient of $I.$
From this point of view, the PDE smoothes functions as it reduces
the $L^{2}((-\infty,0],\mu$)-norm of the derivative in $x$)

\begin{align*}
\int_{-\infty}^{0}(u_{x}+\varepsilon v_{x})^{2}d\mu & =\int_{-\infty}^{0}(u_{x}+\varepsilon v_{x})^{2}e^{2rx}dx\\
 & =\int_{-\infty}^{0}u_{x}^{2}e^{2rx}dx+2\varepsilon\int_{-\infty}^{0}{u_{x}v_{x}e^{2rx}dx}+\mathcal{O}(\varepsilon^{2})
\end{align*}
At the same time, integration by parts (and applying Neumann conditions
on the boundary) results in 
\[
2\varepsilon\int_{-\infty}^{0}u_{x}v_{x}e^{2rx}dx=-2\varepsilon\int_{-\infty}^{0}v\frac{\partial}{\partial x}\left(u_{x}e^{2rx}\right)dx.
\]
We have 
\[
\frac{\partial}{\partial x}\left(u_{x}e^{2rx}\right)=u_{xx}e^{2rx}+2u_{x}re^{2rx}=(u_{xx}+2ru_{x})e^{2rx}
\]
Therefore, 
\begin{align*}
-2\varepsilon\int_{-\infty}^{0}v\frac{\partial}{\partial x}\left(u_{x}e^{rx}\right)dx & =-2\varepsilon\int_{-\infty}^{0}v(u_{xx}+2ru_{x})e^{2rx}dx\\
 & =-4\varepsilon\int_{-\infty}^{0}v(\frac{1}{2}u_{xx}+ru_{x})d\mu
\end{align*}
By $L^{2}-$duality (or Cauchy-Schwarz), this quantity is made as
small as possible when 
\[
v=\frac{1}{2}u_{xx}+ru_{x}.
\]

That is, $\frac{{1}}{2}u_{xx}+ru_{x}$ is the gradient flow that maximally
smoothes income in the sense of maximally decreasing the present value
of the sum of $u^{2}$. 

\section{Conclusion}

We examine a classic public finance question from a new perspective
and propose an averaging rule based on a small set of assumptions. 

Anticipating potential criticism, we now address some of the broad
issues with this approach. First, the assumptions that we used, while
reasonable, are certainly not the only ones one can use and, hence,
derive a different averaging and smoothing rule. A good parallel to
make is a discussion in the mathematical imaging literature that examines
how various sets of assumptions generate different smoothing mechanisms.
Moreover, the focus there is exactly the one we take here -- how
a small set of assumptions generate reasonable results and what the
consequences are of relaxing or changing some of those. In particular,
we believe it could be quite desirable to have the same question addressed
from various different perspectives and see what kind of averaging
methods may arise from completely different sets of axioms. A fascinating
question is whether the universal appearance of the Gaussian in the
two-sided case has an analogous \textquotedbl universal\textquotedbl{}
averaging scheme. Both the half-sided Gaussian and the exponential
distribution are natural candidates. Second, the question of the practicality
of the results. While the exponential weighting, Gaussian and the
explicit form of the density of the reflected Brownian motion for
the intermediate case are simple mathematically, it is more difficult
(with the exception of the exponential case) to imagine them being
implemented in practice. The abstract formulation of the problem and
the explicit solution that we derive allow, however, to both precisely
state and solve the question rather than rely on the perhaps more
useful heuristics. At the same time, with increased digitization one
can imagine that some of the theoretical insights presented here to
be implemented in practice. This is the main point of an excellent
discussion of the broad range of practical implementation topics of
the theoretical taxation literature, including income averaging, in
Jacobs (2017).

\newpage{}

\section*{References}

\[
\]

Alvarez, Luis, Frederic Guichard, Pierre-Louis Lions, and Jean-Michel
Morel. \textquotedbl Axioms and fundamental equations of image processing.\textquotedbl{}
Archive for Rational Mechanics and Analysis 123, no. 3: 199-257. 1993.

Aubert, Gilles, and Pierre Kornprobst. Mathematical problems in image
processing: partial differential equations and the calculus of variations.
Vol. 147. Springer Science \& Business Media. 2006.

Bargain, Olivier, Alain Trannoy, and Adrien Pacifico. \textquotedbl The
Impact of Tax Frequency: Theoretical and Empirical Investigations\textquotedbl .
Working paper. 2017.

Batchelder, Lily L. \textquotedbl Taxing the poor: Income averaging
reconsidered.\textquotedbl{} Harvard Journal on Legislation. 40: 395.
2003.

Brown, Robert G. \textquotedbl Exponential smoothing for predicting
demand.\textquotedbl{} In Operations Research, vol. 5, no. 1, pp.
145-145. 1957.

Buchanan, Neil H. \textquotedbl The Case Against Income Averaging.\textquotedbl{}
Virginia Tax Review. 25: 1151. 2005.

Diamond, Peter. \textquotedbl Taxes and pensions.\textquotedbl{}
Southern Economic Journal 76, no. 1 (2009): 2-15

Fagerstr{\"o}m, Daniel. \textquotedbl Temporal scale spaces.\textquotedbl{}
International Journal of Computer Vision 64, no. 2-3: 97-106. 2005.

Feller, William. \textquotedbl Diffusion processes in one dimension.\textquotedbl{}
Transactions of the American Mathematical Society 77, no. 1: 1-31.
1954.

Glynn, Peter W., and Rob J. Wang. \textquotedbl On the rate of convergence
to equilibrium for reflected Brownian motion.\textquotedbl{} Queueing
Systems 89, no. 1-2 (2018): 165-197.

Gordon, Daniel V., and Jean-Francois Wen. \textquotedbl Tax penalties
on fluctuating incomes: estimates from longitudinal data.\textquotedbl{}
International Tax and Public Finance 25, no. 2: 430-457. 2018.

Harrison, J. Michael. Brownian models of performance and control.
Cambridge University Press, 2013.

Heathcote, Jonathan, Kjetil Storesletten, and Giovanni L. Violante.
\textquotedbl Consumption and labor supply with partial insurance:
An analytical framework.\textquotedbl{} American Economic Review 104,
no. 7: 2075-2126. 2014.

Holt, Charles C. \textquotedbl Averaging of income for tax purposes:
Equity and fiscal-policy considerations.\textquotedbl{} National Tax
Journal 2, no. 4 : 349-361. 1949.

Holt, Charles C. \textquotedblleft Forecasting seasonals and trends
by exponentially weighted moving averages\textquotedblright{} Office
of Naval Research. Research Memorandum No. 52., 1957. and reprinted
in Holt, Charles C. \textquotedbl Forecasting seasonals and trends
by exponentially weighted moving averages.\textquotedbl{} International
journal of forecasting 20, no. 1: 5-10. 2004.

Huggett, Mark, and Juan Carlos Parra. \textquotedbl How well does
the US social insurance system provide social insurance?.\textquotedbl{}
Journal of Political Economy 118, no. 1: 76-112. 2010.

Jacobs, Bas. \textquotedbl Digitalization and Taxation\textquotedbl ,
in: Sanjeev Gupta, Michael Keen, Alpa Shah, and Genevieve Verdier
(eds), Digital Revolutions in Public Finance, Washington-DC: International
Monetary Fund, Ch. 2. 2017.

Karlin, Samuel. Total positivity. Vol. 1. Stanford University Press.
1968.

Kapicka, Marek. \textquotedbl{} Quantifying the Welfare Gains from
History Dependent Income Taxation.\textquotedbl{} ADEMU Working paper
series. 2017.

Kolmogorov, Andrei. \textquotedbl {\"U}ber die analytischen Methoden
in der Wahrscheinlichkeitsrechnung.\textquotedbl{} Mathematische Annalen
104, no. 1: 415-458. 1931.

Koenderink, Jan J. \textquotedbl The structure of images.\textquotedbl{}
Biological cybernetics 50, no. 5: 363-370. 1984.

Lindeberg, Tony. \textquotedbl On the axiomatic foundations of linear
scale-space.\textquotedbl{} In Gaussian scale-space theory, pp. 75-97.
Springer, Dordrecht. 1997.

Lindeberg, Tony. \textquotedbl Generalized Gaussian scale-space axiomatics
comprising linear scale-space, affine scale-space and spatio-temporal
scale-space.\textquotedbl{} Journal of Mathematical Imaging and Vision
40, no. 1: 36-81. 2011.

Lindeberg, Tony. Scale-space theory in computer vision. Vol. 256.
Springer Science \& Business Media. 2013.

Lindeberg, Tony, and Daniel Fagerstr{\"o}m. \textquotedbl Scale-space
with casual time direction.\textquotedbl{} In European Conference
on Computer Vision, pp. 229-240. Springer, Berlin, Heidelberg. 1996.

Mintz, Jack M., and Thomas A. Wilson. \textquotedbl Saving the future:
restoring fairness to the taxation of savings.\textquotedbl{} Commentary-CD
Howe Institute 176. 2002.

Mirrlees, James A., and Stuart Adam. Dimensions of tax design: the
Mirrlees review. Oxford University Press, 2010.

Pauwels, Eric J., Luc J. Van Gool, Peter Fiddelaers, and Theo Moons.
\textquotedbl An extended class of scale-invariant and recursive
scale space filters.\textquotedbl{} IEEE Transactions on Pattern Analysis
and Machine Intelligence 17, no. 7: 691-701. 1995.

Peetre, Jaak. \textquotedbl Une caracterisation abstraite des operateurs
differentiels.\textquotedbl{} Mathematica Scandinavica: 211-218. 1959.

Salden, Alfons H., Bart M. ter Haar Romeny, and Max A. Viergever.
\textquotedbl Linear scale-space theory from physical principles.\textquotedbl{}
Journal of Mathematical Imaging and Vision 9, no. 2: 103-139. 1998.

Saez, Emmanuel. \textquotedbl Optimal income transfer programs: intensive
versus extensive labor supply responses.\textquotedbl{} The Quarterly
Journal of Economics 117, no. 3 1039-1073. 2002.

Skorokhod, Anatoliy V. \textquotedbl Stochastic equations for diffusion
processes in a bounded region.\textquotedbl{} Theory of Probability
\& Its Applications 6, no. 3: 264-274. 1961.

Skorokhod, Anatoliy V. \textquotedbl Stochastic equations for diffusion
processes in a bounded region. II.\textquotedbl{} Theory of Probability
\& Its Applications 7, no. 1: 3-23. 1962.

Sch{\"o}nberg, Isaac Jacob. \textquotedbl On variation-diminishing integral
operators of the convolution type.\textquotedbl{} Proceedings of the
National Academy of Sciences of the United States of America 34, no.
4, p. 164-169. 1948.

Simons, Henry C. Personal income taxation: The definition of income
as a problem of fiscal policy. Chicago University, Chicago. 1938.

Steinerberger, Stefan and Aleh Tsyvinski. ``Tax mechanisms and gradient
flows''. National Bureau of Economic Research Working paper, No.
w25821. 2019.

Stroock, Daniel W. Partial differential equations for probabilists.
Cambridge Univ. Press. 2008.

Vickrey, William. \textquotedbl Averaging of income for income-tax
purposes.\textquotedbl{} Journal of Political Economy 47, no. 3: 379-397.
1939.

Weickert, Joachim, Seiji Ishikawa, and Atsushi Imiya. \textquotedbl On
the history of Gaussian scale-space axiomatics.\textquotedbl{} In
Gaussian scale-space theory, pp. 45-59. Springer, Dordrecht. 1997.
\end{document}